\documentclass[11pt]{amsart} 

\usepackage{times,amsrefs,enumerate}

\usepackage[margin=3cm]{geometry}

\usepackage[cmtip,curve,all]{xy}

\let\To=\to
\input diagxy

\renewcommand{\to}{\To}

\usepackage{array}

\newcolumntype{C}{>{$}c<{$}}
\newcolumntype{L}{>{$}l<{$}}
\newcolumntype{R}{>{$}r<{$}}

\usepackage{tikz-cd}

\newtheorem{theorem}{Theorem}[section]
\newtheorem{proposition}[theorem]{Proposition}
\newtheorem{lemma}[theorem]{Lemma}
\newtheorem{corollary}[theorem]{Corollary}

\DeclareMathOperator{\gh}{gh}
\DeclareMathOperator{\pa}{p}
\DeclareMathOperator{\End}{End}
\DeclareMathOperator{\ad}{ad}
\DeclareMathOperator{\pr}{pr}

\newcommand{\A}{\mathcal{A}}
\newcommand{\F}{\mathcal{F}}
\newcommand{\CC}{\mathbb{C}}
\newcommand{\Z}{\mathbb{Z}}
\newcommand{\R}{\mathcal{R}}
\newcommand{\RR}{\mathbb{R}}
\newcommand{\CO}{\mathcal{O}}
\newcommand{\p}{\partial}
\newcommand{\half}{\tfrac{1}{2}}
\newcommand{\quart}{\tfrac{1}{4}}

\newcommand{\Q}{\mathsf{Q}}

\renewcommand{\[}{\{\!\!\{}
\renewcommand{\]}{\}\!\!\}}

\renewcommand{\P}{\mathsf{P}}
\newcommand{\X}{\mathsf{X}}
\newcommand{\E}{\mathsf{E}}
\newcommand{\C}{\mathsf{C}}
\newcommand{\s}{\mathsf{s}}
\newcommand{\ts}{\tilde\s}
\renewcommand{\t}{\mathsf{t}}
\renewcommand{\S}{\mathsf{S}}
\renewcommand{\SS}{\mathbb{S}}
\newcommand{\T}{\mathsf{T}}

\newcommand{\f}{\mathsf{f}}
\newcommand{\CV}{\mathcal{V}}
\newcommand{\eps}{\varepsilon}
\newcommand{\G}{\mathsf{G}}

\newcommand{\W}{\mathsf{W}}

\renewcommand{\frame}{\omega}
\newcommand{\om}{\omega}
\newcommand{\Om}{\Omega}

\renewcommand{\AA}{\mathbb{A}}
\newcommand{\BB}{\mathbb{B}}

\newcommand{\tint}{{\textstyle\int}}

\title{The spinning particle with curved target}

\author{Ezra Getzler}

\address{Department of Mathematics, Northwestern University, Evanston,
Illinois, USA}

\email{getzler@northwestern.edu}

\date{}

\begin{document}

\maketitle
\begin{abstract}
  We extend our previous calculation of the BV cohomology of
  the spinning particle with a flat target to the general case, in
  which the target carries a non-trivial pseudo-Riemannian metric and
  a magnetic field.
\end{abstract}

\section{Introduction}

Unlike in other models which have been investigated, the BV cohomology
of the spinning particle with a flat target is nontrivial in all
negative degrees \cite{cohomology}, raising the question of whether
our understanding of the BV formalism is incomplete. In this paper, we
show that these results extend to the spinning particle with general
target, in which the target carries a non-trivial pseudo-Riemannian
metric carrying a possibly non-zero magnetic field.

The quantum theory associated to this model is familiar to
mathematicians as the Dirac operator on a manifold; the magnetic field
corresponds to twisting by a complex line bundle.

The BV formalism associates to a solution of the classical master
equation
\begin{equation*}
  \left\{ \tint S\,dt,\tint S\,dt \right\} = 0
\end{equation*}
a vector field $\s$ on the space of fields, given by the explicit
formula
\begin{equation}
  \label{s}
  \s = \sum_i (-1)^{\pa(\Phi_i)}\sum_{\ell=0}^\infty \left( \p^\ell
    \biggl( \frac{\delta S} {\delta\Phi_i} \biggr)  \frac{\p\ }
    {\p(\p^\ell\Phi_i^+)} + \p^\ell \biggl(
    \frac{\delta S}{\delta\Phi_i^+} \biggr)
    \frac{\p\ }{\p(\p^\ell\Phi_i)} \right) .
\end{equation}
In Section 2, we show in complete generality that the classical master
equation implies that $\s^2=0$. Our proof of this statement employs a
modified Batalin-Vilkovisky (anti)bracket which differs from the usual
one by a total derivative, and satisfies the graded Jacobi formula on
densities, without the need for any total derivative corrections. This
bracket was introduced (in the ungraded setting) by Soloviev
\cite{Soloviev} and applied to BV geometry in \cite{darboux}.

In Section 3, we derive the master action of the spinning
article. With these technical details out of the way, we calculate the
BV cohomology of the spinning particle in Section~4: it turns out that
the description is essentially identical to the special case discussed
in \cite{cohomology}.

P. Mn\"ev has remarked (private communication) that the model
considered in this paper may also be constructed by the method of
Alexandrov et al.\ \cite{AKSZ}. We discuss this reformulation of the
theory at the end of Section~3.

In Section~4, we discuss the quantum master equation for the spinning
particle. One expects neither anomalies nor renormalization in a
quantum mechanical system, and this is confirmed by our calculations:
there is a potential contribution to the full action at one-loop
(which in fact vanishes for typical regularization schemes), and no
higher-loop contributions.

\section{The Batalin-Vilkovisky formalism}

In the Batalin-Vilkovisky formalism, there are fields $\Phi_i$, of
ghost number $\gh(\Phi_i)\in\Z$ and parity $\p(\Phi_i)\in\Z/2$, along
with the corresponding antifields $\Phi_i^+$, of ghost number
$\gh(\Phi_i^+)=1-\gh(\Phi_i)$, and parity
$\pa(\Phi_i^+)=1-\pa(\Phi_i)$.

We focus on the classical BV formalism for a single independent
variable $t$ (classical mechanics). Let $\p$ denote the total
derivative with respect to $t$. Denote by $\A^j$ the superspace of all
differential expressions in the fields and antifields with
$\gh(S)=j$. The sum $\A$ of the superspaces $\A^j$ for $j\in\Z$ is a
graded superalgebra. A vector field is a graded derivation of the
graded superalgebra $\A$. An example is the total derivative $\p$.

We denote by $\p_{k,\Phi}:\A^j\to\A^j$ the partial derivative
\begin{equation*}
  \p_{k,\Phi} = \frac{\p\ }{\p(\p^k\Phi)} ,
\end{equation*}
and by $\delta_{k,\Phi}:\A^j\to\A^j$ the higher Euler operators of
Kruskal et al.\ \cite{KMGZ}
\begin{equation*}
  \delta_{k,\Phi} = \sum_{\ell=0}^\infty \tbinom{k+\ell}{k} \, (-\p)^\ell
  \p_{k+\ell,\Phi} .
\end{equation*}
When $k=0$, $\delta_{0,\Phi}=\delta_\Phi$ is the classical variational
derivative.

A vector field $\xi$ is called \textbf{evolutionary} if it commutes
with $\p$. Such a vector field is determined by its value on the
fields $\Phi$ and the antifields $\Phi^+$:
\begin{align*}
  \xi &= \sum_i \sum_{k=0}^\infty \left( \p^k \bigl( \xi(\Phi_i)
        \bigr) \, \p_{k,\Phi_i} + \p^k \bigl( \xi(\Phi_i^+) \bigr) \,
        \p_{k,\Phi_i^+} \right) \\
      &= \sum_i \pr \left( \xi(\Phi_i) \frac{\p\ }{\p\Phi_i}
        + \xi(\Phi_i^+ ) \frac{\p\ }{\p\Phi_i^+} \right) .
\end{align*}
The operation $\pr$ is called \textbf{prolongation}.

The \textbf{Soloviev bracket} is defined by the formula
\begin{multline*}
  \[ f , g \] = \sum_i (-1)^{(\pa(f)+1)\pa(\Phi_i)} \\
  \sum_{k,\ell=0}^\infty \left( \p^\ell \bigl( \p_{k,\Phi_i}f \bigr)
    \p^k \bigl( \p_{\ell,\Phi_i^+}g \bigr) + (-1)^{\pa(f)} \p^\ell
    \bigl( \p_{k,\Phi_i^+}f \bigr) \p^k \bigl( \p_{\ell,\Phi_i}g
    \bigr) \right) .
\end{multline*}
It is proved in \cite{darboux} that the bracket $\[f,g\]$
satisfies the following equations:
\begin{description}
\item[skew symmetry]
  $\[ f , g \] = - (-1)^{(\pa(f)+1)(\pa(g)+1)} \[ g , f \]$ \\[-7pt]
\item[Jacobi]
  $\[ f , \[ g , h \] \] = \[ \[ f , g \] , h \] +
  (-1)^{(\pa(f)+1)(\pa(g)+1)} \[ g , \[ f , h \] \]$ \\[-7pt]
\item[linearity over $\p$]
  $\[ \p f , g \] = \[ f , \p g \] = \p \[ f , g \]$
\end{description}

The superspace $\F=\A/\p\A$ of functionals is the graded quotient of
$\A$ by the subspace $\p\A$ of total derivatives. The image of
$f\in\A$ in $\F$ is denoted by $\int f\,dt$, and the bracket induced
on $\F$ by the Soloviev bracket is denoted
\begin{equation*}
  \textstyle
  \{ \int f\,dt , \int g\,dt \} .
\end{equation*}
This bracket may also be written directly in terms of the variational
derivatives:
\begin{equation*}
  \{ \tint f\,dt , \tint g\,dt \} = \sum_i
  (-1)^{(\pa(f)+1)\pa(\Phi_i)}
  \int \left( \bigl( \p_{\Phi_i}f \bigr) \bigl( \p_{\Phi_i^+}g \bigr)
    + (-1)^{\pa(f)} \bigl( \p_{\Phi_i^+}f \bigr) \bigl( \p_{\Phi_i}g
    \bigr) \right) dt .
\end{equation*}

The Batalin-Vilkovisky formalism for classical field theory involves
the selection of a solution of the classical master
equation
\begin{equation*}
  \textstyle \{ \int S\,dt , \int S\,dt \} = 0 ,
\end{equation*}
where $S\in\A^0$ is an element with $\pa(S)=0$. When the antifields
are set to zero, the expression $S(\Phi,0)$ is the classical action.

Stated in terms of the Soloviev bracket, the classical master equation
becomes the equation
\begin{equation}
  \label{master}
  \half \[ S , S \] = \p\tilde{S} ,
\end{equation}
where $\tilde{S}\in\A^1$ is an element with $\pa(\tilde{S})=1$.

\begin{proposition}
  \label{fk}
  The differential operator $\ad(f)=\[f,-\]$ is given by the formula
  \begin{equation*}
    \ad(f) = \sum_{k=0}^\infty \p^k \f_k ,
  \end{equation*}
  where $\f_k$ is the sequence of evolutionary vector fields
  \begin{equation*}
    \f_k = \sum_i (-1)^{(\pa(f)+1)\pa(\Phi_i)} \pr\left( \bigl(
      \delta_{k,\Phi_i}f \bigr) \frac{\p\ }{\p\Phi_i^+} +
      (-1)^{\pa(f)} \bigl( \delta_{k,\Phi_i^+}f \bigr) \frac{\p\ }
      {\p\Phi_i} \right) .
  \end{equation*}
\end{proposition}
\begin{proof}
  We see that
  \begin{multline*}
    \sum_{j,k,\ell=0}^\infty (-1)^j \, \tbinom{k+j}{k} \, \p^k \left(
      \p^{\ell+j} \left( \p_{k+j,\Phi}f \right) \p_{\ell,\Phi^+}g \right) \\
    \begin{aligned}
      &= \sum_{i,j,k,\ell=0}^\infty (-1)^j \, \tbinom{k+j}{k} \,
      \tbinom{k}{i} \, \p^{\ell+i+j} \left( \p_{k+j,\Phi}f \right) \,
      \p^{k-i} \left( \p_{\ell,\Phi^+}g \right) \\
      &= \sum_{i,j,\ell,m=0}^\infty (-1)^j \, \tbinom{m}{m-i-j} \,
      \tbinom{i+j}{j} \, \p^{\ell+i+j} \left( \p_{m,\Phi}f \right) \,
      \p^{m-i-j} \left( \p_{\ell,\Phi^+}g \right) \\
      &= \sum_{\ell,m=0}^\infty \p^\ell \left( \p_{m,\Phi}f \right) \,
      \p^m \left( \p_{\ell,\Phi^+}g \right) ,
    \end{aligned}
  \end{multline*}
  and the analogous equation holds with the roles of $\Phi$ and
  $\Phi^+$ exchanged. Summing over the fields $\Phi_i$, the result
  follows.
\end{proof}

Given a solution of the classical master equation \eqref{Master}, the
functions $S$ and $\tilde{S}$ give rise to the evolutionary vector
fields $\s_k$ and $\ts_k$ respectively, where the vector field $\s_0$
is the vector field $\s$ of \eqref{s}. Define the vector fields
\begin{equation*}
  \sigma_k = \ts_k - \half \sum_{\ell=0}^{k+1} [\s_\ell,s_{k-\ell+1}] .
\end{equation*}

\begin{lemma}
  \begin{equation}
    \label{ss}
    \s^2 = \sum_{k=0}^\infty \p^{k+1} \sigma_k
  \end{equation}
\end{lemma}
\begin{proof}
  The equation $(d+\ad(\S))^2=0$ implies that
  $\ad(S)^2+\p\ad(\tilde{S})=0$. In other words,
  \begin{equation*}
    \s^2 + \sum_{k=0}^\infty \sum_{\ell=0}^{k+1}
    \p^{k+1}\s_\ell\s_{k-\ell+1}
    = \sum_{k=0}^\infty \p^{k+1}\ts_k ,
  \end{equation*}
  which proves the result after a little rearrangement.
\end{proof}

We can now prove the main result of this section.
\begin{theorem}
  If $S$ is a solution of the classical master equation
  \eqref{master}, then the associated vector field $\s$ satisfies the
  equation $\s^2=0$.
\end{theorem}
\begin{proof}
  The idea of the proof is that whereas the right-hand side is a
  vector field of \eqref{ss} is a vector field, the left-hand side is
  a differential operator of degree $>1$. Taking the symbols of both
  sides, we see that the symbol of this differential operator must
  vanish.

  We now prove by downward induction in $k$ that the vector fields
  $\sigma_k$ vanish. Let $K$ be the largest integer such that
  $\sigma_K$ is nonzero. (For the solution of the classical master
  equation associated to a first-order field theory, $K=1$.) Let
  $\Phi$ be one of the fields of the theory having $\pa(\Phi)=0$ (that
  is, a bosonic field), and take the $(K+2)$-fold commutator of both
  sides of \eqref{ss} with $\Phi$.  The differential operator $\s^2$
  is a vector field, so the left-hand side vanishes, while the
  right-hand side equals
  \begin{equation*}
    (K+2)! \, (\p\Phi)^{K+1} \sigma_K(\Phi) .
  \end{equation*}
  It follows that $\sigma_K(\Phi)=0$.

  Next, we take the commutator with the antifield $\Phi^+$ followed by
  the $(K+1)$-fold commutator with $\Phi$: again, the left-hand side
  vanishes, while the right-hand side equals
  \begin{equation*}
    (K+1)! \, (\p\Phi)^K \Bigl( (\p\Phi) \sigma_K(\Phi^+) +
    (K+1) (\p\Phi^+) \sigma_K(\Phi) \Bigr) .
  \end{equation*}
  We have already shown that the second of the two term vanishes, and
  we conclude that $\sigma_K(\Phi^+)=0$.

  The vanishing of $\sigma_K(\Phi)$ and $\sigma_K(\Phi^+)$ may be
  proved for fields $\Phi$ with $\pa(\Phi)=1$ (fermionic fields) by
  exchanging the r\^oles of $\Phi$ and its antifield $\Phi^+$ in the
  above argument. In this way, we see that $\sigma_K=0$. Arguing by
  downward induction, we conclude that $\sigma_k=0$ for all $k\ge0$,
  proving the theorem.
\end{proof}

The vector field $\s$ induces a differential on $\F$, whose cohomology
$H^*(\F,\s)$ is the Batalin-Vilkovisky cohomology of the model. By
Proposition~\ref{fk}, $\s$ equals the differential $\ad(S)$ induced by
taking Soloviev bracket with the solution $S$ of the classical master
equation.

We may calculate the BV cohomology groups $H^*(\F,\s)$ using the
complex
\begin{equation*}
  \CV^j = \A^j \oplus \tilde{\A}^{j+1}\,\eps ,
\end{equation*}
where
\begin{equation*}
  \tilde{\A}^j =
  \begin{cases}
    \A^0/\CC , & j=0 , \\
    \A^j , & \text{$j\ne0$} ,
  \end{cases}
\end{equation*}
with differential
\begin{equation*}
  d \bigl( f + g\,\eps \bigr) = (-1)^{\pa(g)} \, \p g .
\end{equation*}
The symbol $\eps$ is understood to have odd parity and ghost number
$-1$, so that the parities of the superspace $\tilde{\A}^{j+1}$ are
reversed in $\CV^j$. This complex is a shifted differential graded Lie
algebra, with respect to the extension of the Soloviev bracket to
$\CV$:
\begin{equation*}
  \[ f_0 + g_0\,\eps , f_1 + g_1\,\eps \] = \[ f_0 , f_1 \] + \[ f_0
  , g_1 \] \, \eps + (-1)^{\pa(f_1)+1} \, \[ g_0 , f_1 \] \, \eps .
\end{equation*}
The differential satisfies
\begin{equation*}
  d \[ a , b\] = \[ da , b \] + (-1)^{\pa(a)+1} \[ a , db \] .
\end{equation*}

\begin{lemma}
  \label{S}
  If $\int S\,dt\in\F$ is a solution of the classical master equation
  \eqref{master}, then
  \begin{equation*}
    \S = S + \tilde{S}\,\eps \in \CV^0
  \end{equation*}
  is a solution of the master equation
  \begin{equation}
    \label{Master}
    d\S + \half \[ \S , \S \] = 0 .
  \end{equation}
\end{lemma}
\begin{proof}
  Applying the operator $\ad(S)$ to both sides of \eqref{master}, we
  see that
  \begin{equation*}
    \half \[ S , \[ S , S \] \] = \[S,\p\tilde{S}\] =
    \p\[S,\tilde{S}\] ,
  \end{equation*}
  and hence that $\[S,\tilde{S}\]=0$.
\end{proof}

For example, the Poisson structure of the KdV hierarchy (Dickey
\cite{Dickey}; cf.\ \cite{darboux}) gives a solution of the classical
master equation \eqref{Master} with $\gh(S)=-2$ instead of $0$, and
$\gh(\eps)=1$ instead of $-1$:
\begin{equation*}
  \S = x^+\p^3x^+ + xx^+\p x^+ + x^+\p x^+\p^2x^+\,\eps .
\end{equation*}

The differentials $d+\s$ and $d+\ad(\S)$ on $\CV^*$ are equivalent, by
the following proposition.
\begin{proposition}
  Let $P$ be the automorphism of $\CV^*$ defined by the formula
  \begin{equation*}
    P ( f + g\eps ) = f + g \, \eps + (-1)^{\pa(f)} \sum_{k=0}^\infty
      \p^k\s_{k+1}f \, \eps .
  \end{equation*}
  Then the differentials $d+\ad(S)$ and $d+\s$ on $\CV$ are related by
  the equation
  \begin{equation*}
    d + \ad(\S) = P(d+\s)P^{-1} .
  \end{equation*}
\end{proposition}
\begin{proof}
   Written out in full, we have
   \begin{align*}
     (d+\ad(\S)) \bigl( f + g\eps \bigr)
     &= \[ S , f \] +(-1)^{\pa(g)}\,\p g + \left( (-1)^{\pa(f)} \, \[
       \tilde{S} , f \] + \[ S , g \] \right) \eps \\
     &= \sum_{k=0}^\infty \p^k\s_kf +(-1)^{\pa(g)}\,\p g + 
       \sum_{k=0}^\infty \left( (-1)^{\pa(f)} \, \p^k\ts_kf +
       \p^k\s_kg \right) \eps .
   \end{align*}
   We see that
   \begin{align*}
     (d+\ad(S))P(f+g\eps) &= \s f + (-1)^{\pa(g)}\,\p g +
     \sum_{k=0}^\infty \left( (-1)^{\pa(f)} \p^k \bigl( \sigma_kf -
       \s_{k+1}\sf \bigr) + \p^k\s_kg \right) \eps \\
     &= \s f + (-1)^{\pa(g)}\,\p g + \sum_{k=0}^\infty \left(
       (-1)^{\pa(f)+1}\,\p^k\s_{k+1}\sf + \p^k\s_kg \right) \eps \\
     &= P(d+\s)(f+g\eps) ,
   \end{align*}
   where on the second line, we have used the vanishing of the vector
   fields $\sigma_k$.
 \end{proof}

\section{The classical master equation for the spinning particle in
  curved target}

In this section, we construct the solution of the classical master
equation associated to the spinning particle in a curved target.

Let $\RR^d$ be a vector space with constant pseudo-metric
$\eta_{ab}=\eta(e_a,e_b)$. The target of the spinning particle is an
open subset $U$ of $\RR^d$, carrying a Riemannian pseudo-metric
$g_{\mu\nu}=g(\p_\mu,\p_\nu)$ with the same signature as $\eta$. Let
$g^{\mu\nu}=g(dx^\mu,dx^\nu)$ be the metric induced by $g$ on the
tangent bundle. In other words,
\begin{equation*}
  g_{\mu\lambda} g^{\lambda\nu} = \delta^\nu_\mu .
\end{equation*}
Similarly, let $\eta^{ab}e_a\otimes e_b$ be the pseudo-metric induced on
$(\RR^d)^*$ by $\eta$.

We will represent the pseudo-metric $g^{\mu\nu}$ by a moving frame
$\frame^a=\frame_\mu^a\,dx^\mu$. Geometrically speaking, a moving
frame is an isometry between the trivial bundle $U\times\RR^d$ with
constant pseudo-metric $\eta$ and the tangent bundle of
$U$. Equivalently, the one-forms $\{\frame^a\}$ satisfy the equation
\begin{equation*}
  g(\frame^a,\frame^b) = \eta^{ab} ,
\end{equation*}
or
\begin{equation*}
  g_{\mu\nu} = \eta_{ab} \frame_\mu^a \frame_\nu^b .
\end{equation*}
We denote by $\frame^\mu_a$ the inverse of $\frame_\mu^a$, in the
sense that
\begin{equation*}
  \frame_\mu^a\frame^\mu_b=\delta^a_b .
\end{equation*}
We may use the frame $\om_\mu^a$ and its inverse $\om^\mu_a$ to
exchange contravariant and covariant indices $\mu$ with upper and
lower internal indices $a$: for example, $A_a=\om^\mu_a A_\mu$.

The physical fields of the spinning particle (fields of ghost number
$0$) are as follows:
\begin{enumerate}[a)]
\item the position $x^\mu$, which is a field of even parity taking
  values in $U$;
\item fields $p_a$ and $\theta^a$, respectively of even and odd
  parity;
\item the graviton $e$ and gravitino $\psi$, respectively even and odd.
\end{enumerate}
In addition, the model has ghosts $c$ and $\gamma$ (fields of ghost
number $1$), corresponding respectively to diffeomorphism in the
independent variable $t$ and local supersymmetry, which are
respectively odd and even.

The connection one-form
$\om^a{}_b=\om_\mu{}^a{}_bdx^\mu\in\Omega^1(U,\End(\RR^d))$ is a
matrix of one-forms on $U$ characterized in terms of the frame
$\om_\mu^a$ by two conditions: it is \textbf{skew-symmetric}
\begin{equation*}
  \om^b{}_a = - \eta_{a\tilde{a}} \eta^{b\tilde{b}}
  \om^{\tilde{a}}{}_{\tilde{b}} ,
\end{equation*}
and \textbf{torsion-free}, that is, satisfies the first Cartan
structure equation
\begin{equation*}
  d\frame^a + \om^a{}_b \wedge \frame^b = 0 .
\end{equation*}
Written in terms of components, this equation becomes
\begin{equation*}
  \p_\mu\frame_\nu^a - \p_\nu\frame_\mu^a + \om_\mu{}^a{}_b
  \frame_\nu^b - \om_\nu{}^a{}_b \frame_\mu^b =
  0 .
\end{equation*}

The curvature
$R^a{}_b=\half R_{\mu\nu}{}^a{}_bdx^\mu dx^\nu\in\Omega^2(U,\End(\RR^d))$
is a skew-symmetric matrix of two-forms defined by the second Cartan
structure equation
\begin{equation*}
  d\om^a{}_b + \om^a{}_c \wedge \om^c{}_b = R^a{}_b .
\end{equation*}
Written in terms of components, this equation reads
\begin{equation*}
  \p_\mu\om_\nu{}^a{}_b - \p_\nu\om_\mu{}^a{}_b + \om_\mu{}^a{}_c
  \om_\nu{}^c{}_b - \om_\nu{}^a{}_c \om_\mu{}^c{}_b =
  R_{\mu\nu}{}^a{}_b .
\end{equation*}
We will need the Bianchi identities for the curvature
$R_{\mu\nu a b}$: the antisymmetrizations of the expressions
$\frame_\lambda^a R_{\mu\nu ab}$ and
\begin{equation*}
  \p_\lambda R_{\mu\nu ab} + \om_\lambda{}^c{}_a R_{\mu\nu cb} -
  \om_\lambda{}^c{}_b R_{\mu\nu ca}
\end{equation*}
in the indices $\{\lambda,\mu,\nu\}$ vanish.

We also introduce a magnetic potential (connection one-form)
\begin{equation*}
  A = A_\mu dx^\mu \in \Omega^1(U)
\end{equation*}
on $U$, with associated field-strength (curvature) $F=dA$, or in terms
of components,
\begin{equation*}
  F_{\mu\nu} = \p_\mu A_\nu - \p_\nu A_\mu .
\end{equation*}

We now turn to the construction of the solution $\S$ of the classical
master equation associated to the moving frame $\frame^a_\mu$ and
magnetic field $A_\mu$. In all of our calculations, the antifield
$x^+_\mu$ enters via the expression
\begin{equation*}
  X^+_a = \om^\mu_a \bigl( x^+_\mu + \om_\mu{}^b{}_c \, p_b p^{+c} -
  \om_\mu{}^b{}_c \, \theta^+_b \theta^c \bigr) .
\end{equation*}

\begin{lemma}
  \label{QG}
  Let $\Sigma\in\A^0$ and $\G\in\A^{-2}$ be given by the formulas
  \begin{align*}
    \Sigma &= \bigl( p_\mu + A_\mu \bigr) \p x^\mu - \half
             \bigl( \eta_{ab} \theta^a \p\theta^b + \om_{\mu ab} \, \p
             x^\mu \theta^a \theta^b \bigr) + \p e^+ c + \p\psi^+\gamma \\
    \G &= X^+_a p^{+a} + \quart p^{+a} p^{+b} \theta^c \theta^d
         R_{abcd} + \half p^{+a} p^{+b} F_{ab} - \half \eta^{ab}
         \theta^+_a \theta^+_b + c^+e + \gamma^+ \psi .
  \end{align*}
  Then $\[\Sigma,\Sigma\]=\[\G,\G\]=0$, and $\[\Sigma,\G\]=\T$, where
  \begin{align*}
    \T &= - x^+_\mu \p x^\mu + \p p^{+a} p_a + \half \bigl(
      \p\theta^+_a \theta^a - \theta^+_a \p\theta^a \bigr)  \\
    &\quad + \p e^+ e + \p c^+ c + \p\psi^+\psi + \p\gamma^+\gamma +
      \p \bigl( A_a p^{+a} \bigr) .
  \end{align*}
\end{lemma}
\begin{proof}
  We may decompose both $\Sigma$ and $\G$ into two parts, the first of
  which only involves the fields $\{x^\mu,p_\mu,x^+_\mu,p^{+\mu}\}$,
  and the second of which involves the remaining fields:
  \begin{align*}
    \Sigma_0 &= \bigl( p_\mu + A_\mu \bigr) \p x^\mu -
           \half \bigl( \eta_{ab} \theta^a \p\theta^b + \om_{\mu ab}
           \, \p x^\mu \theta^a \theta^b \bigr) , \\
    \Sigma_1 &= \p e^+ c + \p\psi^+\gamma , \\
    \G_0 &= X^+_a p^{+a} + \quart p^{+a} p^{+b} \theta^c \theta^d
           R_{abcd} + \half p^{+a} p^{+b} F_{ab} - \half
           \eta^{ab} \theta^+_a \theta^+_b , \\
    \G_1 &= c^+e + \gamma^+ \psi .
  \end{align*}
  The formulas $\[\Sigma,\Sigma\]=\[\G_1,\G_1\]=0$ and
  \begin{equation*}
    \[ \Sigma_1 , \G_1 \] = \p e^+ e + \p c^+ c + \p\psi^+\psi +
    \p\gamma^+\gamma
  \end{equation*}
  are easily verified, and it is also clear that
  $\[\Sigma_i,\G_j\]=\[\G_i,\G_j\]=0$ if $i\ne j$.

  The formulas
  \begin{equation*}
    \[\Sigma_0,\G_0\] = -x^+_\mu\p x^\mu+\p p^{+a} p_a + \half \bigl(
    \p\theta^+_a \theta^a - \theta^+_a \p\theta^a \bigr)+ \p
    \bigl( A_a p^{+a} \bigr)
\end{equation*}
and $\{\G_0,\G_0\}=0$ are a consequence of the structure equations and
the Bianchi identities, together with the corresponding equations
$F=dA$ and $dF=0$ for the magnetic potential and its field strength.
\end{proof}

The interest of this result is that $\ad(\T)=\t_0+\p\t_1$ where
$\t_0=\p$ and
\begin{align*}
  \t_1 &= - \bigl( x^+_\mu - p^{+a} \p_\nu A_a \bigr) \frac{\p\
         }{\p x^+_\nu} - \bigl( p_a + A_a \bigr) \frac{\p\ }{\p
         p_a} - \frac12 \left( \theta^+_a \frac{\p\ }{\p\theta^+_a}
         - \theta^a \frac{\p\ }{\p\theta^a} \right) \\
       &\quad - e \frac{\p\ }{\p e} - c \frac{\p\ }{\p c} -
         \psi \frac{\p\ }{\p\psi} - \gamma \frac{\p\ }{\p\gamma} .
\end{align*}

The following proposition gives a method of constructing solutions of
the classical master equation.
\begin{proposition}
  Let $\W\in\A^1$ satisfy the equations $\[\Sigma,\W\]=0$ and
  \begin{equation*}
    \[\[\G,\W\],\W\]=0 .
  \end{equation*}
  Then $\S=\Sigma+\[\G,\W\]+(\W+\t_1\W)\eps$
  is a solution of the classical master equation
  \begin{equation*}
    d\S + \half \[ \S , \S \] = 0 .
  \end{equation*}
\end{proposition}
\begin{proof}
  The proposition is implied by Lemma~\ref{S}, if we can prove the
  equation
  \begin{equation*}
    \half \{ \Sigma + \[\G,\W\] , \Sigma + \[\G,\W\] \] =
    \p(\W+\t_1\W) .
  \end{equation*}
  By the graded Jacobi relation, we see that
  \begin{multline*}
    \half \{ \Sigma + \[\G,\W\] , \Sigma + \[\G,\W\]
    \] = \bigl( d\Sigma + \half \[ \Sigma , \Sigma \] \bigr) +
    \quart \[ \[ \[ \G , \G \] , \W \] , \W \] \\
    - \[ G , \[ \Sigma , \W \] \] - \half \[ \G , \[ \[ \G , \W \] ,
    \W \] \] + \[ \T , \W \] .
  \end{multline*}
  Both terms on the first line vanish by Lemma~\ref{QG}, while the
  first two terms on the second line vanish by hypothesis. The result
  follows from the formula
  \begin{equation*}
    \[ \T , \W \] = \p\W + \p\t_1\W .
    \qedhere
  \end{equation*}
\end{proof}

We now consider the expression
\begin{equation}
  \W = \half \eta^{ab} p_a p_b c + \half F_{ab} \theta^a
    \theta^b c + p_a \theta^a \gamma - e^+ \gamma^2 \in \A^1 .
    \label{W}
\end{equation}
It is clear that $\[\Sigma,\W\]=0$, and a somewhat lengthier
calculation shows that
\begin{equation*}
  \[\[\G,\W\],\W\] \in \A^2
\end{equation*}
vanishes as well. It follows that
\begin{align*}
  \S &= \Sigma + \[ \G , \W \] + (\W+\t_1\W)\eps \\
     &= \left( \frame_\mu^a p_a + A_\mu \right) \p x^\mu - \half
       \bigl( \eta_{ab} \theta^a \p\theta^b + \om_{\mu ab} \p x^\mu
       \theta^a \theta^b \bigr) \\
     &- \half e \left( \eta^{ab} p_a p_b + F_{ab} \theta^a \theta^b
       \right) + \psi p_a \theta^a \\
     &+ \bigl( \p e^+ - \eta^{ab} X^+_a p_b + \half p^{+a} p_b
       \theta^c \theta^d R_a{}^b{}_{cd}
       - p^{+a} p_b F_a{}^b + \theta^+_a \theta^b F^a{}_b +
       \half \om^\lambda_a p^+_a \theta^b \theta^c \nabla_\lambda
       F_{bc} \bigr) c \\
     &+ \bigl( \p\psi^+ - X^+_a \theta^a + \eta^{ab} \theta^+_a
       p_b + 2 e^+ \psi - p^{+a} \theta^b F_{ab} \bigr) \gamma - c^+ \gamma^2  \\
     &  - \bigl( \eta^{ab} p_a p_b c +  \tfrac32 p_a \theta^a \gamma -
       e^+ \gamma^2 - \half F_{ab} \theta^a \theta^b c +
       \eta^{ab} A_a p_b c + A_a \theta^a \gamma \bigr) \eps
\end{align*}
satisfies the classical master equation
$d\S + \half \[ \S , \S \] = 0$. In this equation, we have denoted by
$\nabla F$ the covariant derivative of the two-tensor $F$ with respect
to the Levi-Civita connection $\om_\mu{}^a{}_b$.

\begin{corollary}
  If $\int f\,dt\in\F^k$ is a cocycle in the complex $(\F,\s)$, where
  $\s$ is the vector field associated to the solution $\S$ of the
  classical master equation, then $\int\[\G,f\]\,dt\in\F^{k-1}$ is a
  cocycle in $\F$, called the \textup{\textbf{transgression}} of $f$.
  In particular, the long exact sequence
  \begin{equation*}
    \begin{tikzcd}
      \cdots \arrow{r} & H^{-1}(\A,\s) \arrow{r}{\p} & H^{-1}(\A,\s)
      \arrow{r} & H^{-1}(\F,\s) \arrow[out=-5,in=170,overlay]{lld} & \\
      & H^0(\A/\CC,\s) \arrow{r}{\p} & H^0(\A,\s) \arrow{r} &
      H^0(\F,\s)
      \arrow[out=-5,in=170,overlay]{lld} & \\
      & H^1(\A,\s) \arrow{r}{\p} & H^1(\A,\s) \arrow{r} & H^1(\F,\s)
      \arrow{r} & \cdots
    \end{tikzcd}
  \end{equation*}
  splits, in the sense that the morphisms $\p$ vanish.
\end{corollary}
\begin{proof}
  Since $\[\G,\G\]=0$, we have the equation
  \begin{equation*}
    \[ \Sigma + \[\G,\W\] , \[ \G , f \] \] = \[ \T , f \] - \[ \G
      , \[ \Sigma + \[ \G , \W \] , f \] \] .
  \end{equation*}
  By hypothesis, $\[\Sigma+\[\G,\W\],f\]=\p
  g$ is a total derivative. Thus
  \begin{equation*}
    \[ \Sigma + \[\G,\W\] , \[ \G , f \] \] = \p ( f + \t_1f + g ) .
  \end{equation*}
  Hence $\[\G,f\]$ descends to a cocycle in $\F^{k-1}$. This shows
  that the connecting morphisms $\p$ in the long-exact sequence
  vanish.
\end{proof}

We close this section by showing how to rewrite $S$ as an AKSZ
action. In AKSZ models, the fields may be assembled into differential
forms of homogeneous total degree: in our case, the sum of a $0$-form
of ghost number $k$ and a $1$-form of ghost number $k-1$. These
differential forms are as follows:
\begin{gather*}
  \begin{aligned}
    \mathbf{x}^\mu &= x^\mu + dt\,\[\G,x^\mu\]
    & {\boldsymbol\theta}^a &= \theta^a + dt\,\[\G,\theta^a\] &
    \mathbf{p}_a &= p_a + dt\,\[\G,p_a\]
  \end{aligned} \\
  \begin{aligned}
    \mathbf{c} &= c + dt\,\[\G,c\] &
    {\boldsymbol\gamma} &= \gamma + dt\,\[\G,\gamma\] \\
    \mathbf{e}^+ &= e^+ + dt\,\[\G,e^+\]
    & {\boldsymbol\psi}^+ &= \psi^+ + dt\,\[\G,\psi^+\]
  \end{aligned}
\end{gather*}
The action $S$ is the one-form component of the differential form
\begin{gather*}
  \bigl( \frame_\mu^a(\mathbf{x}) \mathbf{p}_a + A_\mu(\mathbf{x}) -
  \half \om_{\mu ab}(\mathbf{x}) {\boldsymbol\theta}^a
  {\boldsymbol\theta}^b \bigr) d\mathbf{x}^\mu - \half \eta_{ab}
  {\boldsymbol\theta}^a d{\boldsymbol\theta}^b + \mathbf{c}\,
  d\mathbf{e}^+ + {\boldsymbol\gamma}\,d{\boldsymbol\psi}^+ \\
  + \half \eta^{ab} \mathbf{p}_a \mathbf{p}_b \mathbf{c} + \half
  F_{ab}(\mathbf{x}) {\boldsymbol\theta}^a {\boldsymbol\theta}^b
  \mathbf{c} + \mathbf{p}_a {\boldsymbol\theta}^a {\boldsymbol\gamma}
  - \mathbf{e}^+ {\boldsymbol\gamma}^2 ,
\end{gather*}
where we recognize the expressions $\Sigma$ and $\W$ of Lemma~\ref{QG}
and \eqref{W} respectively on the first and second lines. The
resemblance between the action in an AKSZ model and the Chern-Simons
action is clear after changing variables from the field $\mathbf{p}_a$
to the field
\begin{equation*}
  \mathbf{P}_\mu = \frame_\mu^a(\mathbf{x}) \mathbf{p}_a +
  A_\mu(\mathbf{x}) - \half \om_{\mu ab}(\mathbf{x})
  {\boldsymbol\theta}^a {\boldsymbol\theta}^b .
\end{equation*}

\section{Calculation of BV cohomology}

The method of \cite[Section~7]{cohomology} may be used to calculate
the BV cohomology of the spinning particle in the general case. Let
$\CO$ be the ring of functions on the target $U\subset\RR^d$ of the
spinning particle: we may take any of the standard structure rings of
geometry, namely algebraic, analytic or infinitely-differentiable
functions, or even power series. Let $\A$ is the graded polynomial
algebra over $\CO$ generated by the remaining variables of the theory,
namely
\begin{multline*}
  \{\p^\ell x^\mu\}_{\ell>0}\cup \{\p^\ell\theta^a,\p^\ell p_a,\p^\ell
  x^+_\mu,\p^\ell\theta^+_a,\p^\ell p^{a+}\}_{\ell\ge0} \\
  \cup \{\p^\ell e,\p^\ell\psi,\p^\ell
  e^+,\p^\ell\psi^+\}_{\ell\ge0}\cup \{\p^\ell c,\p^\ell\gamma,\p^\ell
  c^+,\p^\ell\gamma^+\}_{\ell\ge0} .
\end{multline*}
Let
\begin{equation*}
  \A_\gamma^*=\A^*\otimes_{\CC[\gamma]}\CC[\gamma,\gamma^{-1}]
\end{equation*}
be the localization of $\A^*$, obtained by inverting the ghost
$\gamma$.

Given a vector $v$ with components $v_a$, define
\begin{equation*}
  \iota(v) = \eta^{ab} v_a \frac{\p\ }{\p\theta^b} .
\end{equation*}
If $f\in\CO$, denote by $\nabla f$ the vector with components
\begin{equation*}
  (\nabla f){}_a = \om^\mu_a \bigl( \p_\mu + \om_\mu{}^b{}_b f \bigr) .
\end{equation*}
We may interpret the function $f$ as representing a section of a line
bundle over $U$ with connection form $\om_\mu{}^b{}_bdx^\mu$.

Let $\Om = \theta^1 \dots \theta^d$. Given a function $f\in\CO$ and
$k\ge0$, consider the following elements of $\A_\gamma^{-k-1}$:
\begin{align*}
  A_k(f) &= (\psi^+)^{k+1} c f \Om \gamma^{-1} , \\
  Z_k(f) &= (k+1) (\psi^+)^k f \Om \gamma^{-1} +
           (\psi^+)^{k+1} c \iota(\nabla f) \Om \gamma^{-1} .
\end{align*}
After application of the BV differential $\s$ to these expressions,
the poles in $\gamma$ cancel, showing that the following expressions
are cocycles in $\A^{-k}$ with respect to the differential $\s$:
\begin{align*}
  \alpha_k(f) &= \s(A_k(f)) , & \zeta_k(f) &= \s(Z_k(f)) .
\end{align*}
Consider also the transgressions of these cocycles:
\begin{align*}
  \tilde\alpha_k(f)
  &= \[ \G , \alpha_{k-1}(f) \] ,
  & \tilde\zeta_k(f)
  &= \[ \G , \zeta_{k-1}(f) \] .
\end{align*}

Let $\R$ be the quotient of the differential graded superalgebra
$\A^*$ by the differential ideal generated by the fields
\begin{equation*}
  \{e,\psi,c\} \cup \{
  x^+_\mu,\theta^+_a,p^{+a},e^+,\psi^+,c^+,\gamma^+\}
\end{equation*}
Denote by $\P_a$, $\Theta^a$, $\X^\mu$ and $\Gamma$ the zero-modes
$\int p_a\,dt$, $\int\theta^a\,dt$, $\int x^\mu\,dt$ and
$\int\gamma\,dt$ respectively. Then $\R$ is the graded superalgebra
\begin{equation*}
  \CO[\Theta^a,\P_a,\Gamma] /
  \bigl(\P_a\Theta^a,\eta^{ab}\P_a\P_b+F_{ab}(\X)\Theta^a\Theta^b,\Gamma^2 \bigr)
\end{equation*}
with differential $\Gamma\Q$, where $\Q$ is the differential operator
\begin{multline}
  \label{Q}
  \Q = \om^\mu_a(\X) \Theta^a \frac{\p\ }{\p\X^\mu}
  + \eta^{ab} \P_a \frac{\p\ }{\p\Theta^b}
  + \om^\mu_c(\X) \om_\mu{}^a{}_b(\X) \Theta^c \left( \P_a \frac{\p\
    }{\p\P_b} - \Theta^b \frac{\p\ }{\p\Theta^a} \right) +
    F_{ab}(\X) \Theta^a \frac{\p\ }{\p\P_b} .
\end{multline}
We denote the element $\Theta^1\dots\Theta^d$ of $\R^0$ by the same
symbol $\Om$ as in $\A^0$.

The map $\xi^0$ from $\CO[\Theta^a,\P_a]$ to $\A^0$ which takes a
function $u$ to the corresponding function $\xi^0(u)$ in the variables
$\{x^\mu,\theta^a,p_a\}$ induces a map from $H^0(\R)$ to
$H^0(\A,\s)$. Observe that $\xi^0(\iota(P)\Om)=-\zeta_0(1)$.

Similarly, the map from $\CO[\Theta^a,\P_a]$ to $\A^1$ which takes a
function $v$ to the element
\begin{equation*}
  \xi^1(v) = \gamma v + c \Q v
\end{equation*}
induces a map from $H^1(\R)$ to $H^1(\A,\s)$. Define the
transgressions of the classes $\xi^0(u)$ and $\xi^1(v)$:
\begin{align*}
  \tilde\xi^{-1}(u) &= \[ \G , \xi^0(u) \] &
  \tilde\xi^0(v) &= \[ \G , \xi^1(v) \] .
\end{align*}

The following theorem has the same form as in the special case where
$g^{\mu\nu}$ is constant and $A_\mu=0$, discussed in
\cite{cohomology}.

\begin{theorem}
  \begin{equation*}
    H^{-k}(\F,\s) =
    \begin{cases}
      \Bigl\{ \int \bigl( \alpha_k(f) + \zeta_k(g) +
      \tilde\alpha_k(\tilde{f}) + \tilde\zeta_k(\tilde{g}) \bigr) \,dt \,
      \Big| \, f,g,\tilde{f},\tilde{g} \in\CO \Bigr\} & k>1 , \\[10pt]
      \Bigl\{ \int \bigl( \tilde{\xi}^{-1}(u) + \alpha_1(f) + \zeta_1(g)
      + \tilde\alpha_1(\tilde{f}) + \tilde\zeta_1(\tilde{g}) \bigr) \,dt
      \,\Big|\, & \\
      \hfill u\in H^0(\R/\CC), f,g,\tilde{f} \in\CO , \tilde{g} \in
      \CO/\CC \Bigr\} & k=1 , \\[10pt]
      \Bigl\{ \int \bigl( \xi^0(u) + \tilde\xi^0(v) + \alpha_0(f) +
      \zeta_0(g) \bigr) \, dt \,\Big|\, & \\
      \hfill u \in H^0(\R) , v\in H^1(\R) , f \in \CO , g\in \CO/\CC
      \Bigr\} & k=0 , \\[10pt]
      \Bigl\{\int\xi^1(v)\,dt\,\Big|\, v\in H^1(\R) \Bigr\} & k=-1 ,
      \\[10pt]
      0 & k<-1 .
    \end{cases}
  \end{equation*}
\end{theorem}

The proof of the theorem follows along the same lines as in Section~7
of \cite{cohomology}. We use the filtration on the complex $(\A^*,\s)$
associated to the parameter $\sigma=0$, which assigns bidegrees to the
fields and their derivatives according to the following table:
\begin{equation*}
  \renewcommand{\arraystretch}{1.4}
  \begin{tabular}{|C|C|C|} \hline \hbox{\ $\Phi$\ } &
  \hbox{\ $(p,q)$\ } &
  \hbox{\ $(p^+,q^+)$\ } \\ \hline
    x & (0,0) & (0,-1) \\
    \theta & (0,0) & (0,-1) \\
    p & (0,0) & (0,-1) \\
    e & (2,0) & (-1,0) \\
    \psi & (2,0) & (-1,0) \\
    c & (2,-1) & (-1,-1) \\
    \gamma & (2,-1) & (-1,-1) \\ \hline
    \end{tabular}
\end{equation*}
Here, $p$ and $p^+$ are the filtration degrees of a field $\Phi$ and
its antifield $\Phi^+$, and $q$ and $q^+$ are the complentary degrees,
such that $\gh(\Phi)=p+q$ and $\gh(\Phi^+)=p^++q^+$. We obtain a
spectral sequence $E^{pq}_r$ such that $E^{pq}_r=0$ if $q>0$, and
$d_r:E^{pq}_r\to E^{p+r,q-r+1}_r$.

It is not \emph{a priori} evident that this spectral sequence
converges.  We will see that, as in \cite{cohomology}, $d_r$ vanishes
for $r\ge3$. its convergence is proved by lifting the cohomology
classes in $E_3$ to the explicit nontrivial cocycles in the original
complex that were introduced above.

The differential $d_0:E_0^{pq}\to E_0^{p,q+1}$ of the initial page
$E_0$ is as follows:
\begin{align*}
  d_0 &= - \bigl( \p p_\mu + \half \p_{[\mu}\om_{\nu]ab} \p x^\nu
        \theta^a \theta^b - \om_{\mu ab} \theta^a \p\theta^b -
        F_{\mu\nu} \p x^\nu \bigr) \frac{\p\ }{\p x^+_\mu} \\
      &\quad + \bigl( \eta_{ab} \p \theta^b + \om_{\mu ab} \p x^\mu
        \theta^b \bigr) \frac{\p\ }{\p\theta^+_a} + \om^a_\mu \p x^\mu
        \frac{\p\ }{\p p^{a+}} \\
      &\quad + \p e^+ \frac{\p\ }{\p c^+} + \p\psi^+ \frac{\p\
        }{\p\gamma^+} - \p c \frac{\p\ }{\p e} + \p\gamma \frac{\p\
        }{\p\psi} .
\end{align*}
It follows that $E_1$ is the tensor product of the algebra $\CO$, with
generators $\X^\mu=\int x^\mu\,dt$, and the free graded commutative
algebra with the following generators:
\begin{equation*}
  \renewcommand{\arraystretch}{1.4}
  \begin{tabular}{|r|c|} \hline
    $\gh$ & \textup{generators} \\ \hline
    $-1$ & $\E^+=\int e^+\,dt$, $\Psi^+=\int \psi^+\,dt$ \\
    $0$ & $\Theta^a=\int \theta^a\,dt$,
          $\P_a=\int p_a\,dt$ \\
    $1$ & $\C=\int c\,dt$, $\Gamma=\int\gamma\,dt$ \\
    \hline
  \end{tabular}
\end{equation*}

The differential $d_1:E_1^{pq}\to E_1^{p+1,q}$ is given by the
formula
\begin{equation*}
  d_1 = - \half \left( \eta^{\mu\nu} \P_\mu \P_\nu + F_{ab}(\X)
    \Theta^a \Theta^b \right) \frac{\p\ }{\p\E^+}
  - \P_\mu \Theta^\mu \frac{\p\ }{\p\Psi^+} .
\end{equation*}
Cohomology classes in $E_2=H^*(E_1,d_1)$ take the general form
\begin{equation*}
  z = [b_0] + \sum_{j>0} \Bigl( [\AA_j(f_j)] + [\BB_j(g_j)] \Bigr) ,
\end{equation*}
where $[b_0]$ is an element of the ring
\begin{equation*}
  \CO[\Theta^a,\P_a,\C,\Gamma]/\bigl( \P_a\Theta^a ,
  \eta^{ab}\P_a\P_b+F_{ab}(\X)\Theta^a\Theta^b \bigr)
\end{equation*}
and
\begin{align*}
  \AA_j(f) &= 2j (\Psi^+)^{j-1} \E^+ f \, \Om - (\Psi^+)^j f
    \iota(\P)\Om , &
  \BB_j(g) &= (\Psi^+)^j f \, \Om ,
\end{align*}
for $f,g\in\CO[\C,\Gamma]$.

The differential $d_2:E_2^{pq}\to E_2^{p+2,q-1}$ is given by the
formula
\begin{equation*}
  d_2 = - \C \Q^2 - \Gamma \Q + \Gamma^2 \frac{\p\ }{\p\C} +
  2\E^+\Gamma \frac{\p\ }{\p\Psi^+} ,
\end{equation*}
where $\Q$ is the differential operator introduced in \eqref{Q}. The
remainder of the proof of the theorem is as in \cite{cohomology}.

\section{The quantum master equation}

The Batalin-Vilkovisky formalism for quantization of a solution $S$ of
the classical master equation involves a series
\begin{equation*}
  \SS = S + \sum_{n=1}^\infty \hbar^n S_n
\end{equation*}
satisfying the quantum master equation
\begin{equation*}
  \half \{ \tint \SS \, dt , \tint \SS \, dt \} + \hbar \tint \Delta
  \SS \, dt = 0 .
\end{equation*}
Expanding in powers of $\hbar$, we see that this amounts to the
sequence of equations
\begin{equation*}
  \s \tint S_{n+1} \, dt + \half \sum_{k=1}^n \{ \tint S_k \, dt ,
  \tint S_{n-k+1} \, dt \} + \hbar \tint \Delta S_n \, dt = 0 ,
  \quad n\ge0 .
\end{equation*}
Here, $\Delta$ is the differential operator
\begin{equation*}
  \Delta \tint f \, dt = \sum_{\Phi_i} (-1)^{\pa(\Phi_i)} \int
  \sum_{k,\ell} \lim_{s\to t} \p^k_s\p^\ell_t \delta(s,t) \frac{\p^2f}
  {\p(\p^k\Phi_i(s))\p(\p^\ell\Phi_i^+(t))} \, dt .
\end{equation*}

The operator $\Delta$ is ill-defined, owing to ultra-violet
divergences. But in the case of the spinning particle, there is a
great simplification, since the only contribution to $\Delta S$ comes
from the terms $-\eta^{ab}X^+_ap_bc$ and $-X^+_a\theta^a\gamma$ of
$S$, and we have
\begin{align*}
  \Delta S &= \Delta \left( - \frame_a^\mu x^+_\mu \bigl( \eta^{ab}
             p_b c + \theta^a \gamma \bigr) + \frame_a^\mu
             \om_\mu{}^c{}_d \bigl( - \eta^{ab} p_cp^{d+}p_b c +
             \theta^+_c\theta^d \theta^a \gamma \bigr) \right) \\
           &= C_\Lambda \bigl( - \p_\mu \frame^\mu_a +
             \frame_a^\mu \om_\mu{}^a{}_b \bigr) \bigl( \eta^{bc}
             p_c c + \theta^b \gamma \bigr) \\
           &= - C_\Lambda \, \s \log \det\bigl( \frame_\mu^a \bigr) .
\end{align*}
where $C_\Lambda$ is a function of the cut-off $\Lambda$. (In fact,
$C_\Lambda$ vanishes in the heat-kernel regularization, since the
world-line $\RR$ is odd-dimensional.) We see that
$S_1=C_\Lambda \log \det\bigl( \frame_\mu^a \bigr)$. Since
$\{S_1,S_1\}$ and $\Delta S_1$ both clearly vanish, we also see that
$S_n=0$, $n>1$. This shows that the solution to the quantum master
equation associated for the spinning particle with curved target is
\begin{equation*}
  \S = S + C_\Lambda \log \det\bigl( \frame_\mu^a \bigr) .
\end{equation*}

\subsection*{Acknowledgements}

  I am grateful to Chris Hull for introducing me to the first-order
  formalism of the spinning particle. This research is partially
  supported by EPSRC Programme Grant EP/K034456/1 ``New Geometric
  Structures from String Theory'' and Collaboration Grant \#243025 of
  the Simons Foundation. This paper was written while the author was a
  member of the program ``Higher Structures in Geometry and Physics''
  at the Max-Planck-Institut in Bonn in 2016.

\end{document}